\documentclass[11pt]{amsart}
\usepackage[margin=1in]{geometry}
\usepackage{amsfonts, float, mathtools}
\usepackage{amssymb}
\usepackage[dvips]{graphics}
\usepackage{epsfig}
\usepackage{etoolbox}
\usepackage[usenames,dvipsnames,svgnames,table,x11names]{xcolor}
\usepackage{microtype}
\usepackage{listings}
\usepackage{euscript}
\usepackage{xcolor}
\usepackage{enumitem}
\usepackage{cite}
\usepackage{subcaption}
\usepackage{hyperref}
\hypersetup{
    colorlinks=true,
    linkcolor=blue,
    filecolor=magenta,      
    urlcolor=green,
}

\usepackage[all]{xy}
\usepackage{multirow}
\usepackage{tikz}
\usepackage{tikz-3dplot}
\usetikzlibrary{decorations.markings}

\usetikzlibrary{matrix,arrows,decorations.pathmorphing,positioning,calc}
\definecolor{Green}{rgb}{0,.8,.4}
\definecolor{Plum}{rgb}{.5,0,1}
\definecolor{cyan}{rgb}{0.50,.9,0.9}
\definecolor{lightblue}{rgb}{0.3, 0.6, .7}
\definecolor{ocre}{RGB}{243,102,25}
\definecolor{mygray}{RGB}{243,243,244}

\usepackage{fancyvrb}
\usepackage{bbm}

\newtheorem{thm}{Theorem}[section]

\newtheorem{lemma}[thm]{Lemma}

\theoremstyle{definition}

\theoremstyle{remark}

\numberwithin{equation}{section}
\newtheorem{theorem}{Theorem}[section]
\newtheorem{counterexample}{Counterexample}[section]

\newtheorem{corollary}{Corollary}[section]

\newcommand{\id}{\mathbbm{1}}

\newcommand{\cA}{{\mathcal A}}

\newcommand{\cG}{{\mathcal G}}
\newcommand{\cH}{{\mathcal H}}

\newcommand{\supp}{\operatorname{supp}}

\newcommand{\caE}{{\mathcal E}}

\newcommand{\caP}{{\mathcal P}}

\newcommand{\caS}{{\mathcal S}}

\newcommand{\caV}{{\mathcal V}}

 \newcommand{\C}{\mathbb{C}}

 \newcommand{\F}{\mathbb{F}}

 \renewcommand\bar\overline

\newcommand{\Tr}{\mathrm{Tr}}

\newcommand{\be}{\begin{equation}}
	\newcommand{\ee}{\end{equation}}
\newcommand{\bea}{\begin{eqnarray}}
	\newcommand{\eea}{\end{eqnarray}}
\newcommand{\beann}{\begin{eqnarray*}}
	\newcommand{\eeann}{\end{eqnarray*}}

\usetikzlibrary{decorations.text}
\title{Long-Range Antiferromagnetic Order in the AKLT Model on Trees and Treelike Graphs}
\author{Thomas Jackson}
\address{\textnormal{(Thomas Jackson)} Department of Mathematics and Center for Quantum Mathematics and Physics, University of California,  Davis, CA  95616-8633, USA}
\address{\textnormal{(Thomas Jackson)} Department of Mathematics, United Arab Emirates University,  Al Ain, Abu Dhabi, UAE}
\date{\today}
\begin{document}
\maketitle
\begin{abstract}
We extend the result of Fannes, Nachtergaele, and Werner on long-range order in the AKLT model on Cayley trees to include various trees and tree-like graphs that obey certain conditions. Our examples split into three cases: Cayley-like tree-like graphs generated by a finite subgraph, for which we have a simple condition; arbitrary trees with a prescribed growth rate of their volume; and bilayer Cayley trees.
\end{abstract}
\begin{section}*{Introduction} 
The AKLT model \cite{aklt} is the first example of a proven Haldane phase \cite{haldane} and of a matrix product state (MPS) \cite{mps} and consequently of a tensor network state (TNS) \cite{orus}. However, basic questions about its behavior in higher dimensions remain unanswered. The ground state has the highly unusual property \cite{aah} that expectations of observables can be given by a classical partition function of a system in the same number of variables. This led Affleck, Kennedy, Lieb, and Tasaki to conjecture that it does not have a unique ground state and that the ground states have antiferromagnetic long-range order (LRO) for high-dimensional lattices and high-degree graphs. This was later shown for the Cayley trees (i.e. Bethe lattices) in \cite{fnw}, with some evidence for further results \cite{pomata21}, but otherwise a proof has remained elusive for general lattices. We extend the result of \cite{fnw} to include various trees and tree-like graphs that obey certain conditions. Our examples split into three cases: Cayley-like tree-like graphs generated by a finite subgraph, for which we have a simple condition; arbitrary trees with a prescribed growth rate of their volume; and bilayer Cayley trees.
This paper is organized as follows:
\par
In \textbf{Section 1} we introduce the AKLT Hamiltonian, prove it is unique up to boundary conditions, give our definition of our infinite-volume ground states, and the graphical language with which we will describe our trees and tree-like graphs.
\par
In \textbf{Section 2} we introduce the graphical language and compute the transfer operator for a single site; we compute explicitly the transfer function and give relevant properties and bounds which will be instrumental for our proof.
\par
In \textbf{Section 3} we prove the result of \cite{fnw} that Cayley trees of degree $5$ or more have degenerate ground state space in infinite volume in two ways which generalize in two separate ways as well.
\par
In \textbf{Section 4} we define the transfer operator on any finite graph with one outgoing index and finitely many ingoing indices. We define our tree-like graphs and show using the graphical picture a condition for a non-unique ground state.
\par
In \textbf{Section 5} we prove for a type of irregular tree that a geometric mean condition on the degree of the vertices is enough to ensure a non-unique ground state.
\par
In \textbf{Section 6} we discuss the bilayer Cayley trees and show the degree $d=5$ bilayer Cayley tree does not have a unique ground state. This shows that local degree determines uniqueness even when the graph has a similar overall structure.
\end{section}
\begin{section}{The AKLT Hamiltonian and ground state on a general tree}
Let $T=(\caV_T,\caE_T)$ be an infinite tree with vertex set $\caV_T$ and edge set $\caE_T$ and distinguished root $\overline{0}\in\caV_T$, with the property that there are no vertices such that $\deg(v)=1$ and $|v-\overline{0}|<\infty$, meaning the leaves are infinitely far away. Define the sequence of vertex sets $\{\caV^n_T\}_{n\geq 0}$ with $\caV^n_T=\{x\in\caV_T:\:|x-\overline{0}|=n\}$ where $|\cdot|$ denotes the familiar graph distance; we call $\caV_T^n$ the $n$th layer of the graph $T$, and call the corresponding $n$-layer tree $T_n=(\bigcup_{i\leq n}\caV_T^i,\bigcup_{i\leq n}\caE_T^i)$ where $\caE_T^n=\{e\in\caE_T:e\cap\caV_T^n\neq\emptyset\}$. We define the AKLT Hamiltonian on $T$ via the formal Hamiltonian \begin{equation}H=\sum_{x\in\caV_T}P^{(\deg(x)/2)}\end{equation} introduced in \cite{aklt}. This Hamiltonian is a sum of projectors, and we show that $\ker(H)\neq\emptyset$ and that the vectors $\psi\in\ker(H)$ are unique up to a choice of boundary condition; moreover, we show the Hamiltonian is frustration-free.
\begin{subsection}{The Weyl Representation and Uniqueness of Finite-Volume States}
We follow the proof of \cite{klt} to construct the AKLT ground state on an arbitrary graph.
Consider the representation of $\mathfrak{su}(2)$ on the space of homogeneous polynomials of degree $(\deg(x)+\deg(y))/2=2S_e$, denoted \begin{equation}\caS=\mathrm{span} \{u^{S_e-k}v^k\:|\:0\leq k\leq 2S_e\}\end{equation}
$\caS$ is a $2S_e+1$ dimensional vector space; note that our spin operators correspond to \begin{align*}
    S^+=S^x+iS^y=u\frac{\partial}{\partial v}\\S^-=S^x-iS^y=v\frac{\partial}{\partial u}\\S^z=\frac{1}{2}\left(u\frac{\partial}{\partial u}-v\frac{\partial}{\partial v}\right)
\end{align*}
And note that the eigenstates of $S^z$ are the polynomials $u^jv^{S_e-j}$ for $0\leq j\leq 2S_e$; we define an inner product so that the polynomials \begin{equation}u^jv^{2S_e-j}\left[\binom{2S_e}{j}(2S_e+1)\right]^{1/2}\end{equation} are an orthonormal basis. This gives a $2S_e+1$-dimensional irreducible representation of $\mathfrak{su}(2)$. The inner product can be made explicit via a change of basis where \begin{equation}u=e^{i\phi/2}\cos(\theta/2)\end{equation} and \begin{equation}v=e^{-i\phi/2}\sin(\theta/2)\end{equation} with $0\leq \theta\leq\pi$ and $0\leq\phi<2\pi$. Then we can define \begin{equation}\langle\psi,\phi\rangle=\int d\Omega \overline{\psi(u,v)}\phi(u,v)\end{equation} where \begin{equation}d\Omega=(4\pi)^{-1}\sin \theta d\theta d\phi\end{equation} It was shown by \cite{aah} that for each local observable $A$ there exists a unique polynomial in $\theta$ and $\phi$ labeled $A(\Omega)$ called the symbol such that \begin{equation}
\langle\psi,A\phi\rangle=\int d\Omega\overline{\psi(u,v)}\phi(u,v)A(\Omega)\end{equation}
At each site $x\in\caV$ we associate the variables $u_x$ and $v_x$ and the spin-$S_x$ subspace is given by the span of homogeneous polynomials of degree $S_x$. We then show the following.
\begin{theorem}\cite{klt}\label{klt}
    Let $\psi\in\cH$ such that $H\psi=0$; then there exists a unique $\phi$ such that $\phi$ is a polynomial in $u_x,v_x$ for $x\in\partial\Lambda$ and such that \begin{equation}\psi=\phi\prod_{(x,y)\in\caE}(u_xv_y-u_xv_y)\end{equation}
    Conversely if the last equation holds then $H\psi=0$
\end{theorem}
\begin{proof}
We first prove two lemmata.
For any edge $(x,y)$ denote the Hilbert space $\cH_x\otimes\cH_y$ and label the subspaces \begin{equation}S_1=\{(u_xv_y-u_yv_x)\phi(u_x,v_x,u_y,v_y)\}\end{equation}
\begin{equation}S_2=\{\psi(u_x,v_x,u_y,v_y)\: :\:P^{2S_e}_{x,y}\psi=0\}\end{equation}
\begin{lemma}\cite{klt}
 $S_1\subseteq S_2$   
\end{lemma}
\begin{proof}
    Define $T^{\alpha}=S^{\alpha}_x+S^{\alpha}_y$ for $\alpha=x,y,z$. Then \begin{equation}T^{\alpha}(u_xv_y-u_yv_x)=0\end{equation} for all $\alpha$ and so \begin{equation}T^{\alpha}\psi=(u_xv_y-v_xu_y)T^{\alpha}\phi\end{equation} and thus \begin{equation}T^{2}\psi=(u_xv_y-v_xu_y)T^{2}\phi\end{equation}
    We note that $\phi$ can be written as a sum of eigenvectors of $T^2$ so that $\phi=\sum_{j}\phi_j$ where $T^2\phi_j=j(j+1)\phi_j$ and thus \begin{equation}T^{2}(u_xv_y-v_xu_y)\phi_j=j(j+1)(u_xv_y-v_xu_y)\phi_j\end{equation} and so $(u_xv_y-v_xu_y)\phi_j$ is an eigenstate of $T^2$; by assumption $\phi_{2S_e}=0$ because it is in the ground state of $P^{2S_e}$ and so $\psi\in S_2$ so we know $S_1\subseteq S_2$.\par
    Next note that $\dim S_2=(2S_e+1)^2-(4S_e+1)=(2S_e)^2$ because we have every eigenspace $(2S_e+1)^2$ minus the ones corresponding to highest spin $4S_e+1$. Note also that $S_1$ is $(2S_e)^2$ dimensional and thus $S_1=S_2$
\end{proof}
\begin{lemma}\cite{klt}
    If $\psi\in\cH=\bigotimes_k\cH_k$ satisfies $P^{2S_e}_{(x,y)}\psi=0$ then $\psi=(u_xv_y-u_yv_x)\phi$; the converse also holds.
\end{lemma}
\begin{proof}
    The space of all $\psi$ satisfying \begin{equation}P^{2S_e}_{(x,y)}\psi=0\end{equation} is $S_2\otimes(\bigotimes_{k\neq x,y}\cH_k)$ and similarly all $\psi$ such that $\psi=(u_xv_y-u_yv_x)\phi$ is \begin{equation}S_1\otimes\left(\bigotimes_{k\neq x,y}\cH_k\right)\end{equation} and thus from the previous lemma we have that the two conditions are the same.
\end{proof}
Now let $\psi\in\cH$ such that $H\psi=0$, Then $(u_xv_y-v_xu_y)$ is a factor of $\psi$ for every $(x,y)\in\caE$. The polynomial ring these are contained in is a unique factorization domain, so if $A$ and $B$ are factors of $\psi$ then $\psi=AB\gamma$ for some unique $\gamma$. Then $\psi=\phi\prod_{(x,y)\in\caE}(u_xv_y-u_yv_x)$ and the converse also holds. Thus $\psi$ is unique up to a polynomial of the remaining boundary variables $\phi$. 
\end{proof}
\end{subsection}
Now letting $T_n$ be the previously defined increasing and exhaustive sequence of subtrees, we can associate a finite volume state $\psi_n(B)$ of the Hamiltonian defined on $T_n$ which is dependent on the boundary condition $B_n\in M_2^{\otimes |\partial T_n|}$, where $\partial T_n=\{x\in\caV_{n+1}\setminus \caV_{n}\}$. We define the infinite-volume ground states of the model as weak limits of finite volume states where \begin{equation}\omega^n_{B_n}(A)=\langle\psi_n(B_n),A\psi_n(B_n)\rangle\end{equation} where $A\in\cA$ is a quasi-local observable and $B_n$ is a sequence of boundary conditions on $T_n$. We say the infinite volume state $\omega:\cA\to\C$ defined as \begin{equation}\omega(A):=\lim_{n\to\infty}\omega^n_{B_n}(A)\end{equation} is unique if the limit is independent of the sequence of boundary conditions $B_n$. One way to compute whether the infinite volume state is independent of the boundary variables taken in the limit is to compute the spectra of the transfer operators, as in the $1$-dimensional case.\par
In order to understand the AKLT model further, we fully describe the transfer operator on a single site of degree $d$.
\end{section}
\section{AKLT model transfer operator on a single site}
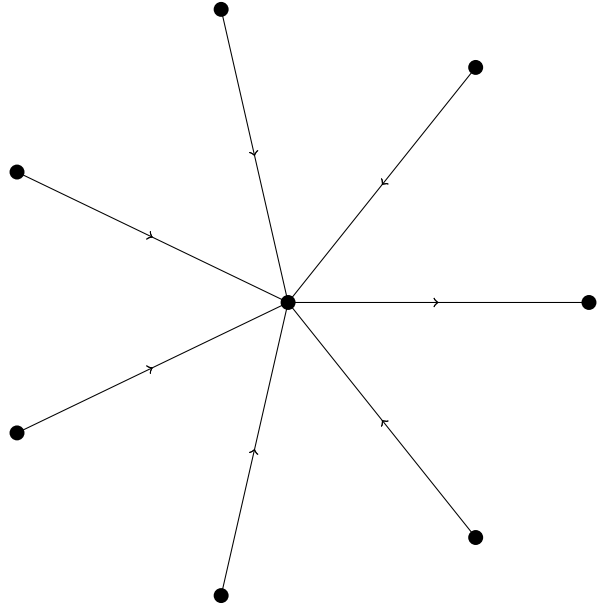
\begin{figure}
    \centering
    \begin{tikzpicture}[
    scale=2, 
    every edge/.style={thick, postaction={decorate}},
    decoration={markings, mark=at position 0.5 with {\arrow{>}}}
]

    \def\r{2} 

    \coordinate (center) at (0, 0);
    \node[circle, fill=black, inner sep=2pt] at (center) {};

    \foreach \i in {1, 2, 3, 4, 5, 6,7} {
        \coordinate (outer\i) at ({360/7*(\i-1)}:\r);
        \node[circle, fill=black, inner sep=2pt] at (outer\i) {};
        
        \ifnum\i=1
            \draw[postaction={decorate, decoration={markings, mark=at position 0.5 with {\arrow{>}}}}] (center) -- (outer\i); 
        \else
            \draw[postaction={decorate, decoration={markings, mark=at position 0.5 with {\arrow{>}}}}] (outer\i) -- (center); 
        \fi
    }

\end{tikzpicture}
    \caption{Vertex of degree $d=7$ with one outgoing index and six ingoing indices}
    \label{d7}
\end{figure}
First we describe each site operator on the lattice. From the valence bond solid picture we have that we can view each physical site as composed of $d$ spin-$\frac{1}{2}$ entities living on the edges projected onto a physical spin-$\frac{d}{2}$ space. Note we have a physical spin-$d/2$ representation $\caP_{(d/2)}$ of $\mathfrak{su}(2)$ of dimension $d+1$ which can be uniquely found as a subspace, meaning \begin{equation}
\caP_{(d/2)}\subset \bigotimes_{i=1}^d\caP_{(1/2)}
\end{equation}
where $\caP_{(1/2)}$ is the $2$-dimensional spin representation.
We have then that the uniqueness of this subspace implies the existence of an SU(2) intertwiner \begin{equation}V^*:\C^{d+1}\otimes (\C^2)^{\otimes d-1}\to \C^2\end{equation} which maps the ingoing space and physical
space $\C^{d+1}\otimes (\C^2)^{\otimes d-1}$ to the outgoing auxiliary space $\C^2$ and satisfies \begin{equation}
    V(S_i)=(\id_{\mathrm{tp}}\otimes S^{d-1}_i+S_i^{\mathrm{tp}}\otimes\id_{d-1})
\end{equation} where $S^d_i$ is the $i$th SU(2) generator in $\caP_{(d/2)}$ and $S^{\mathrm{tp}}_i$ is the $i$th SU(2) generator on $\bigotimes_{i=1}^d\caP_{(1/2)}$. We can construct an associated non-normalized transfer operator $\F:M_2(\mathbbm{C})^{\otimes d-1}\to M_2(\mathbbm{C})$ given by \begin{equation}\F(X\otimes Y)=V^*(X\otimes Y)V\end{equation} where we normalize so that $\F(\id\otimes\id)=c\id$. We will refer to this unnormalized transfer operator for ease of computation in some places. To compute infinite volume expectations on the tree, we will need to know how $\F$ evaluates when
$X=\id$ and $Y$ is arbitrary; in this case we have that the transfer operator breaks down into spin subspaces labeled by $j$ as \begin{equation}\label{unnormalized}\F(\id\otimes Y)=\F(Y)=\sum_{j=-\frac{d-1}{2}}^{\frac{d-1}{2}}W_j Y W_j^*\end{equation} where \begin{equation}S=|\uparrow\rangle\langle \downarrow|-|\downarrow\rangle\langle \uparrow|,\qquad W_j=SP_j\end{equation} where $P_j$ is the projection onto the subspace of spin $j$; note we drop the dependence on the identity since it will be taken for granted. We choose our boundary indices to be written in the Hilbert-Schmidt basis, that is, as sums of tensor products of Pauli matrices, and we will see that we can provide a formula for each term in this basis. We can compute these weights by choosing a basis of unnormalized symmetric states in the following way.\par  Let $|w_k\rangle$ be the symmetric state with $k$ ups so that the total spin is $j=k-(d-1)/2$ (for example $|w_1\rangle$ for $d=3$ would be $|\uparrow\downarrow\rangle+|\downarrow\uparrow\rangle$). On a site of degree $d$ one can write the unnormalized projection as
\begin{equation}P_j'=|\uparrow\rangle\langle w_{k-1}|+|\downarrow\rangle\langle w_k|\end{equation}
We then normalize by total number of terms which is the length of $w_{k-1}$ and $w_k$ added together, which gives $\binom{d-1}{k-1}+\binom{d-1}{k}=\binom{d}{k}$ so that 
\begin{equation}P_j=\frac{1}{\sqrt{\binom{d}{k}}}(|\uparrow\rangle\langle w_{k-1}|+|\downarrow\rangle\langle w_k|)\end{equation} as this ensures the normalization condition $\tilde{\F}(\id)=\id$; note we differentiate our normalized transfer operator $\tilde{\F}$ from our unnormalized transfer operator $\F$. Thus we have \begin{equation}\tilde{\F}(M)=\frac{2}{d+1}\sum_k W_j M W_j^*\end{equation} where $M\in M_2(\mathbbm{C})^{\otimes d-1}$ and $\tilde{\F}:M_2(\mathbbm{C})^{\otimes d-1} \to M_2(\mathbbm{C})$. Note that our normalization now insures \begin{equation}\tilde{\F}(\mathbbm{1}^{\otimes d-1})=\frac{2}{d+1}\sum_{k=0}^{d-1}\frac{\binom{d-1}{k}}{\binom{d}{k}}\mathbbm{1}=\frac{2}{d+1}\sum_{k=0}^{d-1}\frac{d-k}{d}\mathbbm{1}=\frac{2}{d+1}\frac{d+1}{2}\mathbbm{1}=\mathbbm{1}\end{equation} \par In order to know how 
this operator acts on a general matrix we look at how it acts on a tensor product of Pauli matrices. We have explicitly 
\begin{align*}    
\mathbbm{1}=\sigma_0=\begin{bmatrix}
1 & 0\\
0 & 1
\end{bmatrix},\qquad \sigma_1=\begin{bmatrix}
0 & 1\\
1 & 0
\end{bmatrix},\qquad\sigma_2=\begin{bmatrix}
0 & -i\\
i & 0
\end{bmatrix}, \qquad \sigma_3=\begin{bmatrix}
1 & 0\\
0 & -1
\end{bmatrix}
\end{align*} Letting \begin{equation}\label{ki}M=(\sigma_1)^{\otimes k_1}\otimes(\sigma_2)^{\otimes k_2}\otimes (\sigma_3)^{\otimes k_3}\otimes(\mathbbm{1})^{\otimes d-1-k_1-k_2-k_3}\end{equation} where this tensor product is taken in any order we have the following identity
\begin{theorem}
    For our unnormalized transfer operator defined in \ref{unnormalized}, if all the $k_i$'s are even in expression \ref{ki} we have \begin{equation}\F(M)=\frac{1}{k_1+k_2+k_3+1}\frac{\displaystyle\binom{\frac{k_1+k_2+k_3}{2}}{\frac{k_1}{2},\frac{k_2}{2},\frac{k_3}{2}}}{\displaystyle\binom{k_1+k_2+k_3}{k_1,k_2,k_3}}\mathbbm{1}\end{equation} 
\end{theorem}
\begin{proof}
We provide an elementary proof by providing some intuition using the spin-flip operator $\sigma_1$.
First we assume $k:=k_1+k_2+k_3$ is even.
Starting with the case of $M=(\sigma_1)^{\otimes k}\otimes \mathbbm{1}^{\otimes(d-1-k)}$ we note that in order for the term $W_j M W_j^*$ will be nonzero for the number of terms in $W_j$ which have the same spin after flipping the first $k$ of them. This implies that the first $k$ of the total string of length $d$ are an equal number of ups and downs. Thus the number of ways to choose such a configuration with total number of $j$ spin ups is \begin{equation}{\binom{k}{k/2}}\binom{d-k}{j-k/2}\end{equation}
Noting that the normalization of $W_j$ is $\binom{d}{j}$ and summing over all possible values of $j$ we get that \begin{align*}\F(M)=\frac{1}{d+1}\sum_{j\geq k/2}^{d-k/2}\frac{{\binom{k}{k/2}}\binom{d-k}{j-k/2}}{{\binom{d}{j}}}\mathbbm{1}=\frac{1}{d+1}\sum_{j\geq k/2}^{d-k/2}\frac{{\binom{j}{k/2}}\binom{d-j}{k/2}}{{\binom{d}{k}}}\mathbbm{1}=\frac{1}{d+1}\frac{\binom{d+1}{k+1}}{\binom{d}{k}}\mathbbm{1}=\frac{1}{k+1}\mathbbm{1}\end{align*}
where we have rewritten the sum on the second line and used Vandermonde's identity to get a closed form for the sum, and divided by two since we have two total outputs $\langle\uparrow|$ and $\langle\downarrow|$.
A similar calculation shows that for $k$ odd we have \begin{equation}\F(M)=-\frac{1}{k+2}\sigma_1\end{equation}\par
By invariance these numbers will be the same for $\sigma_2$ and $
\sigma_3$ with the same output for odd with a different index.\par
For the case of  \begin{equation}M=(\sigma_1)^{\otimes k_1}\otimes(\sigma_2)^{\otimes k_2}\otimes (\sigma_3)^{\otimes k_3}\otimes(\mathbbm{1})^{\otimes d-1-k_1-k_2-k_3}\end{equation}
we make the same counting argument and notice that the number of configurations is \begin{equation}\sum_{j\geq (k_1+k_2+k_3)/2}^{d-(k_1+k_2+k_3)/2}\binom{k_1}{k_1/2}\binom{k_2}{k_2/2}\binom{k_3}{k_3/2}\binom{d-k_1-k_2-k_3}{j-k_1/2-k_2/2-k_3/2}\end{equation}
and so by normalization we have \begin{align*}\F(M)=\frac{1}{d+1}\sum_{j\geq (k_1+k_2+k_3)/2}^{d-(k_1+k_2+k_3)/2}\frac{\binom{k_1}{k_1/2}\binom{k_2}{k_2/2}\binom{k_3}{k_3/2}\binom{d-k_1-k_2-k_3}{j-k_1/2-k_2/2-k_3/2}}{{\binom{d}{j}}}\mathbbm{1}\\=\frac{1}{d+1}\sum_{j\geq (k_1+k_2+k_3)/2}^{d-(k_1+k_2+k_3)/2}\frac{\binom{\frac{k_1+k_2+k_3}{2}}{k_1/2,k_2/2,k_3/2}\binom{d-j}{k_1/2+k_2/2+k_3/2}\binom{j}{k_1/2,k_2/2,k_3/2}}{\binom{k}{k_1,k_2,k_3}\binom{d}{k_1+k_2+k_3}}\id\\=\frac{1}{k_1+k_2+k_3+1}\frac{\binom{\frac{k_1+k_2+k_3}{2}}{k_1/2,k_2/2,k_3/2}}{\binom{k_1+k_2+k_3}{k_1,k_2,k_3}}\mathbbm{1}\end{align*}
 where again we have used Vandermonde's identity
 \end{proof}
\begin{theorem}\label{comb}
    Letting  $k'_i:=k_i+k_i\mod 2$, if at most one $k_i$ is such that $k_i\neq k_i'$ then \begin{equation}\F(M)=\frac{1}{k'_1+k'_2+k'_3+1}\frac{\displaystyle\binom{\frac{k'_1+k'_2+k'_3}{2}}{\frac{k'_1}{2},\frac{k'_2}{2},\frac{k'_3}{2}}}{\displaystyle\binom{k'_1+k'_2+k'_3}{k'_1,k'_2,k'_3}}\prod_i (-\sigma_i)^{k'_i-k_i}\end{equation}
\end{theorem}
\begin{proof}
    Similar as above, using Vandermonde's identity.
\end{proof}
Note that because of normalization the coefficients in for the transfer operator are not dependent on the degree of the site $d$, only on the number of occupied edges.
\par
We will choose a simple boundary condition to help understand how the concatenations of these operators work on trees, for which the product boundary condition is an invariant subspace of $\F$.
Denote the vector $\mathbf{x}=[x_1,x_2,x_3]$ with $||\mathbf{x}||\leq 1$ and denote $\boldsymbol{\sigma}=[\sigma_1,\sigma_2,\sigma_3]$ We will explore the action of the transfer operator on the vector-parametrized boundary condition \begin{equation}B(\mathbf{x})=(\mathbbm{1}+\mathbf{x}\cdot\boldsymbol{\sigma})^{\otimes (d-1)}\end{equation}
We evaluate \begin{align*}\Tr(\F(B(\mathbf{x})))=\sum_{\substack{ k_i\mathrm{even}}}\binom{d-1}{k_1,k_2,k_3,d-1-k}\cdot \frac{\displaystyle\binom{\frac{k_1+k_2+k_3}{2}}{\frac{k_1}{2},\frac{k_2}{2},\frac{k_3}{2}}}{\displaystyle\binom{k_1+k_2+k_3+1}{k_1,k_2,k_3}}\cdot x_1^{k_1}x_2^{k_2}x_3^{k_3}\\=\sum_{\substack{ k_i\:\mathrm{even}}}\frac{1}{k+1}\binom{d-1}{k}\cdot \displaystyle\binom{\frac{k_1+k_2+k_3}{2}}{\frac{k_1}{2},\frac{k_2}{2},\frac{k_3}{2}}\cdot x_1^{k_1}x_2^{k_2}x_3^{k_3}=\sum_{k\:\mathrm{even}}\frac{1}{k+1}\binom{d-1}{k}\cdot||\mathbf{x}||^k\\=\frac{\left(1+||\mathbf{x}||\right)^d-\left(1-||\mathbf{x}||\right)^d}{2d\cdot||\mathbf{x}||}\end{align*}
We will call this function $f_d(\mathbf{x})=f_d(||\mathbf{x}||)$ for convenience.
\par
Similarly we have \begin{align*}\Tr(\F(B(\mathbf{x}))\cdot\sigma_i)=-\sum_{\substack{k_1+k_2+k_3=k \\ k_1\mathrm{odd}, k_2,k_3\mathrm{even}}}\binom{d-1}{k_1,k_2,k_3,d-2-k}\cdot \frac{\displaystyle\binom{\frac{k_1+1+k_2+k_3}{2}}{\frac{k_1+1}{2},\frac{k_2}{2},\frac{k_3}{2}}}{\displaystyle\binom{k_1+k_2+k_3}{k_1+1,k_2,k_3}}\cdot x_1^{k_1}x_2^{k_2}x_3^{k_3}\\=-\sum_{\substack{k_1+k_2+k_3=k\neq 0 \\ k_i\mathrm{even}}}\binom{d-1}{k_1-1,k_2,k_3,d-k}\cdot \frac{\displaystyle\binom{\frac{k_1+k_2+k_3}{2}}{\frac{k_1}{2},\frac{k_2}{2},\frac{k_3}{2}}}{\displaystyle\binom{k_1+k_2+k_3}{k_1,k_2,k_3}}\cdot x_1^{k_1-1}x_2^{k_2}x_3^{k_3}\\=-\sum_{\substack{k_1+k_2+k_3=k\neq 0 \\ k_i\mathrm{even}}}\frac{k_1}{d}\binom{d}{k_1,k_2,k_3,d-k}\cdot \frac{\displaystyle\binom{\frac{k_1+k_2+k_3}{2}}{\frac{k_1}{2},\frac{k_2}{2},\frac{k_3}{2}}}{\binom{k_1+k_2+k_3}{k_1,k_2,k_3}}\cdot x_1^{k_1-1}x_2^{k_2}x_3^{k_3}\\=-\frac{1}{d}\frac{\partial}{\partial x_i}{f_{d+1}}(\mathbf{x})=-\frac{x_i}{d||\mathbf{x}||}f'_{d+1}(||\mathbf{x}||)\end{align*}
Lastly we note that our trace normalized operator $\tilde{\F}$ evaluates as \begin{align}||\tilde{\F}(B(\mathbf{x}))-\mathbbm{1}||=-F_d(||\mathbf{x}||):=\left\lvert\frac{f'_{d+1}(||\mathbf{x}||)}{df_d(||\mathbf{x}||)}\right\rvert\\=\frac{(d+1)\cdot||\mathbf{x}||\cdot\left((1+||\mathbf{x}||)^d+(1-||\mathbf{x}||)^d\right)-\left((1+||\mathbf{x}||)^{d+1}-(1+||\mathbf{x}||)^{d+1}\right)}{(d+1)||\mathbf{x}||\cdot \left((1+||\mathbf{x}||)^d-(1-||\mathbf{x}||)^d\right)}\\=\frac{1}{d+1}\left(d\cdot\coth(d\cdot \tanh^{-1}(||\mathbf{x}||))-\frac{1}{||\mathbf{x}||}\right)\end{align}
\subsection{The transfer function $F_d(t)$}
We analyze the transfer function $F_d(t)$
\begin{theorem}\label{function}
    For $d\geq 2$ we have the following properties of \begin{equation}F_d(t)=-\frac{1}{d+1}\left(d\cdot\coth(d\cdot\tanh^{-1}(t))-\frac{1}{t}\right)\end{equation}
    \begin{enumerate}
    \item $F_d(t)$ is odd and $F_d(0)=0$
    \item $||F_d(t)||<1$
    \item $F_d(t)$ is analytic
    \item $F_d'(t)>0$
    \item $F_d''(t)<0$ on $(0,1]$ and $F_d''(t)>0$ on $[-1,0)$.
    \item $F_d'(0)=\frac{1-d}{3}$
    \item \begin{equation}-\left(\frac{d-1}{3}\right)\cdot t\leq F_d(t)\leq-\left(\frac{3}{(d-1)t}+1\right)^{-1}\end{equation}
    \end{enumerate}
\end{theorem}
    \begin{proof}
    All of these except the right-hand side the last inequality are given in, or follow quickly from those in \cite{fnw}; the right-hand side of g) follows from the continued fraction expansion for $\coth(t)$ \cite{wall}. Note that \begin{equation}d\coth(dx)=\frac{1}{x}+\frac{d^2 x}{3+\frac{d^2x^2}{5+\frac{d^2x^2}{7+\frac{d^2 x^2}{\ddots}}}}\end{equation} and $F_d(\tanh(x))\leq F_d(x)$ so
\begin{align*}
    F_d(x)^{-1}\geq F_d(\tanh(x))^{-1}=-\left(\frac{1}{d+1}\left(d\coth(dx)-\coth(x)\right)\right)^{-1}&\\=-(d+1)\left(\frac{d^2 x}{3+\frac{d^2x^2}{5+\frac{d^2x^2}{7+\frac{d^2 x^2}{\ddots}}}}-\frac{x}{3+\frac{x^2}{5+\frac{x^2}{7+\frac{x^2}{\ddots}}}}\right)^{-1}\\\geq-\frac{3}{(d-1) x}-\frac{1}{(d-1)}\frac{d^2x}{5+\frac{d^2x^2}{7+\frac{d^2x^2}{\ddots}}}\geq-\left(\frac{3}{(d-1)x}+1\right)
\end{align*}
    \end{proof}
The last inequality allows us to bound iterates of these functions.
This already provides us with enough to give a non-uniqueness condition for Cayley trees, which duplicates the result of \cite{fnw}.
\begin{section}{Cayley trees: two proof methods}
Our tree in this case is an infinite tree $T^d=\{\caV_{T^d},\caE_{T^d}\}$ where for all $x\in\caV_{T^d}$ we have $\deg(x)=d$. $T^d$ has a distinguished root $\overline{0}$ and we can define the set of transfer operators on our volumes $T^d_n$ as \begin{equation}\tilde{\F}_{n,d}:M_2^{\otimes |\partial T^d_n|}\to M_2^{\otimes |\partial T^d_{n-1}|}\end{equation} where $\tilde{\F}_{n,d}=\tilde{\F}_d^{\otimes (d-1)^{n}}$ where $\tilde{\F}_d: M_2^{\otimes d-1}\to M_2$ is the single site transfer operator analyzed in the previous section. We provide two proofs.\par
\begin{figure}[t]
\begin{tikzpicture}[level distance=1.5cm]
    \node[circle, draw, fill=white, minimum size=5pt, inner sep=0pt] (root) at (0, 0) {0};

    \foreach \i in {1, 2, 3, 4, 5} {
        \node[circle, draw, fill=white, minimum size=5pt, inner sep=0pt] (child\i) at (\i*72:1.5) {};
        \draw (root) -- (child\i);
        
        \foreach \j in {1, 2, 3, 4} {
            \node[circle, draw, fill=white, minimum size=5pt, inner sep=0pt] (grandchild\i\j) at (\i*72 + \j*15:3) {};
            \draw (child\i) -- (grandchild\i\j);
        }
    }

\end{tikzpicture}
\caption{First two layers of Cayley tree of degree $d=5$}
\end{figure}
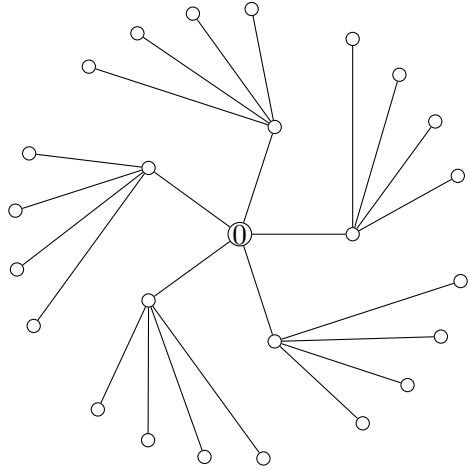
\begin{theorem}
  For $d\geq 5$, the AKLT ground state is not unique.  
\end{theorem}
\begin{proof}[Proof 1]
    Note $\tilde{\F}_{n,d}=\tilde{\F}_d^{\otimes (d-1)^{n}}$ and $\tilde{\F}_d(B(\mathbf{x}))=\id+F_d(||\mathbf{x}||)\cdot\frac{\mathbf{x}}{||\mathbf{x}||}$ if $||\mathbf{x}||\neq 0$ and $\tilde{\F}_d(B(\mathbf(0)))=\id$. Thus if there exists a value $t_0\in[0,1]$ with $F_d(t_0)=\pm t_0$ then $\tilde{\F}_d(B(\mathbf{x}))-\id=\pm\mathbf{x}$ for all $\mathbf{x}$ with $||\mathbf{x}||=t_0$. Defining the limiting boundary conditions $B_n(\mathbf{x})=(\id+\mathbf{x}\cdot\sigma)^{\otimes \partial T_n}$ we have for some $A\in\cA$ with $\supp(A)=\{\overline{0}\}$ we have \begin{align*}
        \langle\psi_n(B_n(\mathbf{x})),A\psi_n(B_n(\mathbf{x}))\rangle=\langle\psi_n(\tilde{\F}^{\circ n}(B_n(\mathbf{x})),A\psi_n(\tilde{\F}^{\circ n}(B_n(\mathbf{x}))\rangle\\=\langle\psi_0(B_1(\pm\mathbf{x})),A\psi_0(B_1(\pm\mathbf{x}))\rangle\neq\langle\psi_0(\id),A\psi_0(\id)\rangle
    \end{align*}
    since $||x||=t_0\neq 0$.
    Thus we need to know if $F_d(t_0)=\pm t_0$ has a solution on $[0,1]$. 
    We have $F_d(0)=0$ and $F_d(t)\leq 0$ on $[0,1]$ and $|F_d(1)|\leq 1$; thus we will have a solution to $F_d(t)=-t$ if $F_d'(0)=\frac{1-d}{3}<-1$ which occurs exactly when $d\geq 5$.
\end{proof}
\begin{proof}[Proof 2]
    Alternatively, notice that the bound $|F_d(x)|\leq\left(\frac{3}{(d-1)x}+1\right)^{-1}$ we have that \begin{equation}|\bigcirc_{i=1}^{n}F_{d}(x)|\geq\left(\frac{1}{x}\prod_{i=1}^{n}\frac{3}{d_i-1}+1+\sum_{k=1}^{n-1}\prod_{i=1}^{k}\frac{3}{d_i-1}\right)^{-1}\end{equation}
    Thus for a fixed $x\in(0,1]$ we have \begin{equation}\lim_{n\to\infty}|\bigcirc_{i=1}^{n}F_{d}(x)|\geq\lim_{n\to\infty}\left(\frac{1}{x}\prod_{i=1}^{n}\frac{3}{d_i-1}+1+\sum_{k=1}^{n-1}\prod_{i=1}^{k}\frac{3}{d_i-1}\right)^{-1}=1-\frac{3}{d-1}\end{equation} which is greater than $0$ when $d\geq 5$. Thus $x_d:=\lim_{n\to\infty}||\tilde{\F}_{T_n}(t_0)-\id||>0$ so the sequences of boundary conditions $B_n(t_0)$ and $B_n(0)$ lead to different infinite volume expectations.
\end{proof}
These two proofs lead to two different ways to generalize the result. The first proof will generalize to graphs which are not necessarily trees, which we can analyze using a graphical or tensor-network expression for the transfer operator; we explain this in the next section. The second proof generalizes to irregular trees with some nice features, which we analyze in Section $5$.
\end{section}
\begin{section}{Tree-like graphs: the graphical picture}
When our graph $T$ is easily constructed via tree-like concatenations of finite volume graphs $\Lambda$ (as the Cayley tree is constructed as a treelike concatenation of the single site graphs) one can leverage the formulae to obtain information on the ground state. We will give a general formula for some of the quantities we have used so far.\par
Let $\Lambda=(\caV_{\Lambda},\caE_{\Lambda})$ be a finite bipartite connected graph and let us define the \textbf{root} $\overline{0}\in\caV_{\Lambda}$ with $\deg(\overline{0})=1$ and the \textbf{boundary} $\partial\Lambda\subset\caV_{\Lambda}$ of $\Lambda$. We construct the \textbf{tree with cell} $\Lambda$
with boundary $\partial\Lambda$ in the following way. Let $d=|\partial\Lambda|+1$ and $\partial\Lambda=\{x_1,...,x_{d-1}\}$ and let $T_d=(\caV_{T_d},\caE_{T_d})$ be the Cayley tree of degree $d$. To each $x\in\caV_{T_d}$ associate a copy $\Lambda^x$ of $\Lambda$. 
Now to each neighbor $N_x=\{y\in\caV_{T_d}\:|\:(y,x)\in\caE_{T_d}\}=\{y^x_1,...,y^x_{d-1}\}$ we add an edge $e_{y^x_k}=(x_k,\overline{0}_{y^x_k})$. Thus we get a treelike graph where to each point in our boundary of $\Lambda$ we have a connection to another copy of $\Lambda$ attached to the root $\overline{0}$.
\begin{figure}[t]
\begin{tikzpicture}[scale=1.5, node distance=1.5cm, every node/.style={circle, fill, inner sep=1.5pt}]

\node (a) at (0, 0) {};
\node[label={[label distance=1mm]315:1}] (b) at (1, 0) {};
\node[label={[label distance=1mm]45:2}] (c) at (1, 1) {};
\node[label={[label distance=1mm]135:3}] (d) at (0, 1) {};
\node[label={[label distance=1mm]135:0}]  (e) at (-.5, -0.5) {};

\draw[thick] (a) -- (b) -- (c) -- (d) -- (a);

\draw[thick] (a) -- (e);

\end{tikzpicture}\hspace{2cm}
\begin{tikzpicture}[scale=1.5, node distance=1.5cm, every node/.style={circle, fill, inner sep=1.5pt}]
\begin{scope}[xshift=3.5cm]

\node (a2) at (0, 0) {};
\node (b2) at (1, 0) {};
\node (c2) at (1, 1) {};
\node (d2) at (0, 1) {};
\node (e2)[label={[label distance=1mm]135:0}] at (-.5, -.5) {};
\draw[thick] (a2) -- (b2) -- (c2) -- (d2) -- (a2);
\draw[thick] (a2) -- (e2);

\node (topa) at (2.5, 1.5) {};
\node (topb) at (2.5, 2.5) {};
\node (topc) at (1.5, 2.5) {};
\node (topd) at (1.5, 1.5) {};
\draw[thick] (topa) -- (topb) -- (topc) -- (topd) -- (topa);
\draw[thick] (c2) -- (topd);

\node (topa) at (-.5, 1.5) {};
\node (topb) at (-.5, 2.5) {};
\node (topc) at (-1.5, 2.5) {};
\node (topd) at (-1.5, 1.5) {};
\draw[thick] (topa) -- (topb) -- (topc) -- (topd) -- (topa);
\draw[thick] (d2) -- (topa);

\node (topa) at (2.5, -.5) {};
\node (topb) at (2.5, -1.5) {};
\node (topc) at (1.5, -1.5) {};
\node (topd) at (1.5, -.5) {};
\draw[thick] (topa) -- (topb) -- (topc) -- (topd) -- (topa);
\draw[thick] (b2) -- (topd);
\end{scope}
\end{tikzpicture}
\caption{First and second layers of the quasi-Cayley tree generated by the cell on the left with root $0$ and boundary $\{1,2,3\}$.}
\label{fig:quasicayley}
\end{figure}
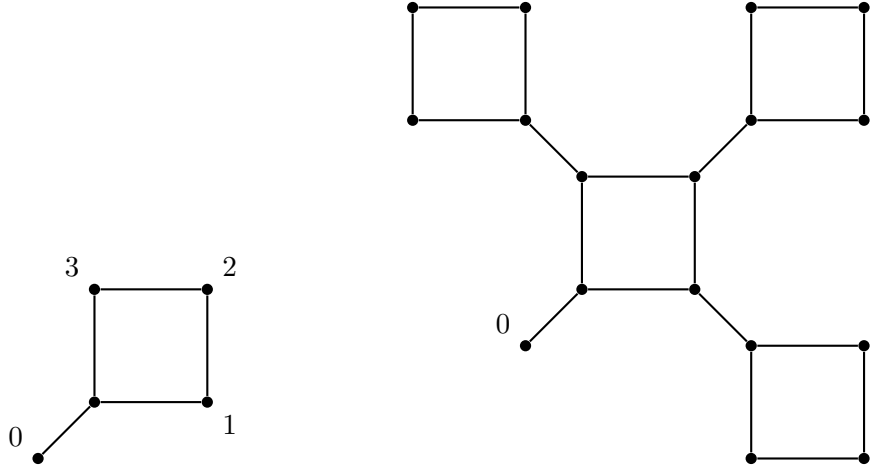
We can define the transfer operator $\tilde{\F}_{\Lambda}: M_2^{\otimes \partial \Lambda}\to M^{\{\overline{0}\}}_2$ and define the modified boundary condition \begin{equation}B(\mathbf{x})=\bigotimes_{y\in\partial \Lambda}(\id +(-1)^{\chi_A(y)}\mathbf{x}\cdot\sigma)\end{equation} where \begin{equation}\chi_A(y)=\begin{cases} 
      1 & y\in \caV_{\Lambda}^A\\
      0 & y\in \caV_{\Lambda}^B
   \end{cases}\end{equation}
where $\caV_{\Lambda}^A$ and $\caV_{\Lambda}^B$ are the bipartition of $\Lambda$.
In order to give an explicit formula for the values of the transfer operator acting on these boundary conditions we will need to describe a set of labeled loop diagrams in $\Lambda$. To each vertex $x_k\in\caV_{\partial\Lambda}$ associate a new vertex $v_k$ and edge $e_k=(v_k,x_k)$ and call $\Lambda'=(\caV_{\Lambda}\cup\bigcup_{i=1}^{d-1}\{v_k\},\caE_{\Lambda}\cup\bigcup_{i=1}^{d-1}\{e_k\})$ the \textbf{augmentation} of $\Lambda$.
Then we define \begin{equation}\cG=\{\Gamma\subset \Lambda'\:|\:\deg{x}=0\:\mathrm{mod}\:2,\:\forall x\in \caV_{\Lambda'}\}\end{equation} and \begin{equation}\cG_k=\{\Gamma\in\cG:\:|\caV_{\Gamma}\cap\{v_1,...,v_{d-1}\}|=k\}\end{equation}
Defining the weight \begin{equation}W(\Gamma)=\prod_{x\in\caV_{\Gamma}}\frac{-1}{\deg(x)+1}\end{equation}
we have the following.
\begin{theorem} Let $\Lambda$ be a finite bipartite graph and $\tilde{\F}_{\Lambda}$ and $B_{\Lambda}(t)$ as above. Then 
    \begin{equation}F_{\Lambda}(t)=-||\tilde{\F}_{\Lambda}(B_{\Lambda}(t))-\id||=\frac{p_{\Lambda}(t)}{q_{\Lambda}(t)}\end{equation} where \begin{equation}p_{\Lambda}(t)=\sum_{k\:\mathrm{odd}}t^k\sum_{\Gamma\in\cG_k}W(\Gamma)\end{equation} and \begin{equation}q_{\Lambda}(t)=\sum_{k\:\mathrm{even}}t^k\sum_{\Gamma\in\cG_k}W(\Gamma)\end{equation}
Moreover $F_{\Lambda}(t)\leq 0$ on $[0,1]$ and $F_{\Lambda}(0)=0$ and $F_{\Lambda}(1)\geq -1$.
\end{theorem}
\begin{proof}
    The operator $\F_{\Lambda}$ is a concatenation of single site operators $\F_x$ for each $x\in C$ where the concatenation of indices is consistent with the edges connecting vertices. Specifically we have \begin{equation}\F_x=T_{e_1,...,e_j}^{e_1',...,e_k'}\end{equation} where $j+k=\deg(x)$ and $T$ is the invariant intertwining tensor that finds $M_{d+1}\subset M_2^{\otimes d}$. Thus the entire transfer operator is such that \begin{equation}\F_{\Lambda}=\sum_{e_i,e_j'\in\caE_{\Lambda}}\bigotimes_{\caV}\F_x\delta_{e_i,e_j'}\end{equation}
    Thus from the formulae above any nonzero term in the transfer operator corresponds to a weighted sum over a subset $\Gamma\subset C$ such that $\deg(x)$ is even for all indices except any in $\partial C$, since every nonzero concatenation must have an entering and leaving index. Thus we can write it as $\sum_{\Gamma\in\cG}W(\Gamma)$ for some weight function $W(\cdot)$. Moreover this weight function factorizes sitewise; we notice that at each site $x\in \Gamma$ if $\deg(x)=2k$ then the weight of the $\F_x$ is given from our above formula as $-\frac{1}{2k+1}=-\frac{1}{\deg(x)+1}$ and is zero if $\deg(x)$ is odd. Then
    every graph $\Gamma$ can be decomposed into connected components which intersect the boundary and those that do not. The key thing to note is that if $\Gamma=\{\gamma_1,...,\gamma_n\}$ are the connected components of $\Gamma$, then there is at most one component $\gamma_j$ intersecting $\overline{0}$ and the boundary. Thus we have for $\Gamma\in\cG_k$ and $k$ even \begin{equation}W(\Gamma)=||\mathbf{x}||^k\prod_{i} W(\gamma_i)\end{equation} where \begin{equation}W(\gamma)=\prod_{x\in\caV_{\gamma}}\frac{-1}{\deg(x)+1}\end{equation}  Then \begin{equation}\Tr(\F_{\Lambda}(B_{\Lambda}(\mathbf{x})))=\sum_{k\:\mathrm{even}}||\mathbf{x}||^k\sum_{\Gamma\in\cG_k}W(\Gamma)=q_{\Lambda}(||\mathbf{x}||)\end{equation} is a function of $||\mathbf{x}||$ only, since every configuration in $\cG_k$ for $k$ even does not touch $\overline{0}$. For the $\sigma_i$ parts of the output, we have that there must be a path from the boundary to the root; thus $\Gamma\in\cG_k$ for $k$ odd, and we have that  \begin{equation}\Tr(\F_{\Lambda}(B_{\Lambda}(\mathbf{x}))\cdot\sigma_i)=\frac{x_i}{||\mathbf{x}||}\sum_{k\:\mathrm{odd}}||\mathbf{x}||^k\sum_{\Gamma\in\cG_k}W(\Gamma)=\frac{x_i}{||\mathbf{x}||}\cdot p_{\Lambda}(||\mathbf{x}||)\end{equation} and so for our trace-normalized operator $\tilde{\F}_{\Lambda}$ we have \begin{equation}||\tilde{\F}_{\Lambda}(B(\mathbf{x}))-\id||=-F_{\Lambda}(||\mathbf{x}||)\end{equation} where \begin{equation}F_{\Lambda}(t)=\frac{p_{\Lambda}(t)}{q_{\Lambda}(t)}\end{equation} with the forms given in the theorem.
    $F_{\Lambda}(0)=0$ then follows since $p_{\Lambda}(t)$ is odd and $q_{\Lambda}(0)\neq 0$ since $\Lambda$ is bipartite; $-1\leq F_{\Lambda}(t)\leq 0$ for all $t\in[0,1]$ follows from the complete positivity of $\F_{\Lambda}$ and trace normalization of $\tilde{\F}_{\Lambda}$. 
\end{proof}
Note our transfer function for the cell in Figure \ref{fig:quasicayley} is $F_{\Lambda}(t)=-\frac{26t}{82+24t^2}$ and $F'_{\Lambda}(0)=-\frac{13}{41}$ which fails the hypothesis of Theorem \ref{quasicayley}.
\begin{corollary}\label{quasicayley}
    For a rooted quasi-tree $T_{\Lambda}$ with generating finite bipartite cell graph $\Lambda$, the AKLT Hamiltonian does not have a unique ground state in the infinite volume limit if \begin{equation}\frac{p_{\Lambda}'(0)}{q_{\Lambda}(0)}=\frac{\sum_{\Gamma\in\cG_1}W(\Gamma)}{\sum_{\Gamma\in\cG_0}W(\Gamma)}<-1\end{equation}
\end{corollary}
\begin{proof}
    As $F_{\Lambda}(t)\leq 0$ on $[0,1]$, $F_{\Lambda}(0)=0$, and $F_{\Lambda}(1)\geq -1$ we have that there is a solution to $F_{\Lambda}(t)=-t$ if $F_{\Lambda}'(0)=\frac{p_{\Lambda}'(0)}{q_{\Lambda}(0)}=\frac{\sum_{\Gamma\in\cG_1}W(\Gamma)}{\sum_{\Gamma\in\cG_0}W(\Gamma)}<-1$. Call this solution $t_{\Lambda}$. Now let $T_{\Lambda}^n$ be the $n$th layer of the tree generated from the cell $\Lambda$; then for an observble $A$ with $\supp(A)=\{\overline{0}\}$ we have that for the finite volume state $\psi_n$ on $T_{\Lambda}^n$ we have \begin{align*}\lim_{n\to\infty}\langle\psi_n(B_{T_{\Lambda}^n}(t_{\Lambda}\cdot\sigma_1)),A\psi_n(B_{T_{\Lambda}^n}(t_{\Lambda}\cdot\sigma_1))\rangle=\lim_{n\to\infty}\langle\psi_0(\tilde{\F}_{T_{\Lambda}^n}(B_{T_{\Lambda}^n}(t_{\Lambda}\cdot\sigma_1))),A\psi_0(\tilde{\F}_{T_{\Lambda}^n}(B_{T_{\Lambda}^n}(t_{\Lambda}\cdot\sigma_1)))\rangle\\=\lim_{n\to\infty}\langle\psi_0(\tilde{\F}_{\Lambda}^{\circ n}(B_{T_{\Lambda}^n}(t_{\Lambda}\cdot\sigma_1))),A\psi_0(\tilde{\F}_{C}^{\circ n}(B_{T_{\Lambda}^n}(t_{\Lambda}\cdot\sigma_1)))\rangle=\langle\psi_0(B_{T_{\Lambda}^0}(\pm t_{\Lambda}\cdot\sigma_1))),A\psi_0(B_{T_{\Lambda}^0}(\pm t_{\Lambda}\cdot\sigma_1))\rangle\\\neq \lim_{n\to\infty}\langle\psi_n(\tilde{\F}_{T_{\Lambda}^n}(B_{T_{\Lambda}^n}(\mathbf{0}))),A\psi_n(\tilde{\F}_{T_{\Lambda}^n}(B_{T_{\Lambda}^n}(\mathbf{0})))\rangle=\langle\psi_0(B_{T_{\Lambda}^0}( \mathbf{0}))),A\psi_0(B_{T_{\Lambda}^0}(\mathbf{0}))\rangle\end{align*}
    Thus the infinite-volume limits given by the sequence $B_{T_{\Lambda}^n}(t_{\Lambda}\cdot\sigma_1)$ and $B_{T_{\Lambda}^n}(\id)$ lead to different expectations for local observables, and so the infinite-volume ground state, as a weak limit of finite-volume states, is not unique.
\end{proof} 
\begin{subsection}{Examples}
\begin{subsubsection}{Decorated Cayley trees}
    We can define a class of trees by adding $g\geq 0$ sites to an existing Cayley tree.
    The function \begin{equation}F_{d,g}(t)=\frac{p_{\Lambda_g}(t)}{q_{\Lambda_g}(t)}=\frac{1}{3^g}F_{d}(t)\end{equation} and we have that \begin{equation}F_{d,g}'(0)=\frac{1}{3^g}\frac{d-1}{3}\end{equation}
    Thus we have the following corollary:
    \begin{corollary}
    Let $\Gamma$ be a Cayley tree of degree $d$ with decoration number $g$. If $d>3^{g+1}+1$ then the AKLT ground state on $\Gamma$ is ordered, and if $d<3^{g+1}+1$ it is disordered and unique.
    \end{corollary}
    \begin{proof}
    From of the simplicity of the tree (including the permutation invariance of the inputs) one can carry out the same analysis as in \cite{fnw} which agrees with the results of \cite{pomata21}.
    \end{proof}
    The uniqueness result will not be recoverable in general.
\end{subsubsection}
\begin{subsubsection}{Rooted Quasi-Cayley trees}
    We can generate a tree out of an underlying graph $\Lambda$ which is also a tree. This gives a condition that is much easier to verify.
    \begin{theorem}
        Let $\Gamma$ be a quasi-Cayley tree generated by a cell graph $\Lambda$ that is also a tree. If \begin{equation}\sum_{\gamma\in \cG_1}W(\gamma)=\sum_{\gamma\in \cG_1}\left(\frac{1}{3}\right)^{|\gamma|}>1\end{equation} then the AKLT model on this tree does not have a unique ground state.
    \end{theorem}
    \begin{proof}
    From the above expression for $F_{\Lambda}(t)$ we have \begin{equation}F_{\Lambda}'(0)=\frac{p_{\Lambda}'(0)}{q_{\Lambda}(0)}=\sum_{\gamma\in \cG_1}W(\gamma)\end{equation} since $q_{\Lambda}(0)$ is the sum \begin{equation}\sum_{\Gamma\in\cG_0}W(\Gamma)=1\end{equation} since $\cG_0=\emptyset$ so $q_{\Lambda}(0)=1$ by the convention that $W(\emptyset)=1$. 
    Thus the condition above
    ensures that $F_{\Lambda}'(0)<-1$ and thus has a solution $t_{\Lambda}$ where $F_{\Lambda}(t_{\Lambda})=-t_{\Lambda}$. From a similar argument as Theorem \ref{quasicayley}, this implies observables at the origin have different ground state expectations.
    \end{proof}
    In words, this condition is saying that the sum of all the paths from the root of the cell $\Lambda$ to any leaves in $\Lambda$, weighted by $(-3)^{1-|\gamma|}$ is greater than $1$.
    The proof above is not enough, however, to show that this occurs for every possible choice of boundary condition; also note that this condition depends not only on the tree $\Lambda$ but also on the root $0$.   
\begin{figure}[t]
\begin{tikzpicture}[
  every node/.style = {circle, draw, fill=white, minimum size=1mm, inner sep=0pt},
  level distance = 1.5cm,
  sibling distance = 2.5cm
]

\node {0} 
  child {node {} 
    child {node {}} 
  }
  child {node {} 
    child {node {}} 
  };

\end{tikzpicture}\hspace{2cm}
\begin{tikzpicture}[
  every node/.style = {circle, draw, fill=white, minimum size=1mm,inner sep=0pt},
  level distance = 1cm,
  level 1/.style={sibling distance=2.5cm},
  level 2/.style={sibling distance=1.5cm}
]

\node (root) {0} 
  child {node {} 
    child {node {}
      child {node {} 
        child {node {}}
      }
      child {node {}}{
        child {node {}}}
    }
  }
  child {node {} 
    child {node {}
      child {node {} 
        child {node {}}
      }
      child {node {}}{
        child {node {}}}
    }
  };

\end{tikzpicture}
\caption{First two layers of the quasi-Cayley tree generated from the tree on the left; note this is also the degree $d=2$ Cayley tree with decoration number $g=1$}
\end{figure}
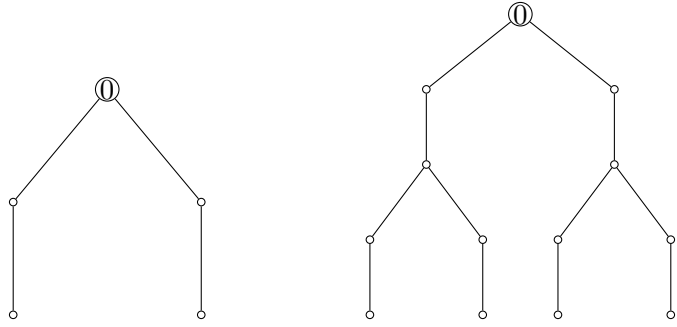
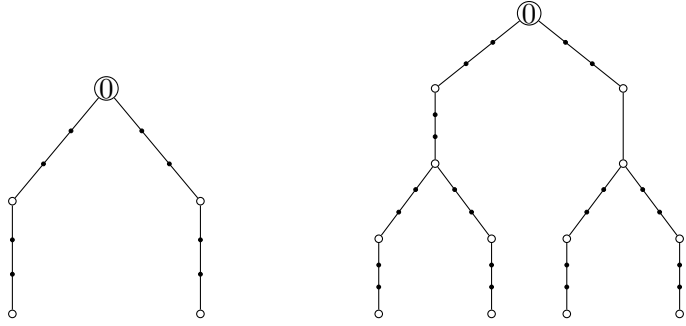
\begin{figure}[t]
    \begin{tikzpicture}[
  every node/.style = {circle, draw, fill=white, minimum size=1mm, inner sep=0pt},
  level distance = 1.5cm,
  sibling distance = 2.5cm,
  intermediate/.style = {fill=black, minimum size=0.5mm, inner sep=0pt} 
]

\node {0} 
  child {node {} 
    child {node {} 
      edge from parent
        node [intermediate, pos=0.33] {} 
        node [intermediate, pos=0.66] {} 
    }
    edge from parent
      node [intermediate, pos=0.33] {} 
      node [intermediate, pos=0.66] {} 
  }
  child {node {} 
    child {node {} 
      edge from parent
        node [intermediate, pos=0.33] {} 
        node [intermediate, pos=0.66] {} 
    }
    edge from parent
      node [intermediate, pos=0.33] {} 
      node [intermediate, pos=0.66] {} 
  };

\end{tikzpicture}
\hspace{2cm}
\begin{tikzpicture}[
  every node/.style = {circle, draw, fill=white, minimum size=1mm, inner sep=0pt},
  intermediate/.style = {fill=black, minimum size=0.5mm, inner sep=0pt}, 
  level distance = 1cm,
  level 1/.style = {sibling distance=2.5cm},
  level 2/.style = {sibling distance=1.5cm}
]

\node {0} 
  child {node {} 
    child {node {}
      child {node {}
        child {node {}
          edge from parent
            node [intermediate, pos=0.33] {} 
            node [intermediate, pos=0.66] {} 
        }
        edge from parent
          node [intermediate, pos=0.33] {} 
          node [intermediate, pos=0.66] {} 
      }
      child {node {}
      child{node {}
        edge from parent
          node [intermediate, pos=0.33] {} 
          node [intermediate, pos=0.66] {} 
          }
        edge from parent
          node [intermediate, pos=0.33] {} 
          node [intermediate, pos=0.66] {} 
      }
      edge from parent
        node [intermediate, pos=0.33] {} 
        node [intermediate, pos=0.66] {} 
    }
    edge from parent
      node [intermediate, pos=0.33] {} 
      node [intermediate, pos=0.66] {} 
  }
  child {node {} 
    child {node {}
      child {node {}
        child {node {}
          edge from parent
            node [intermediate, pos=0.33] {} 
            node [intermediate, pos=0.66] {} 
        }
        edge from parent
          node [intermediate, pos=0.33] {} 
          node [intermediate, pos=0.66] {} 
      }
      child {node {}
      child{node {}
        edge from parent
          node [intermediate, pos=0.33] {} 
          node [intermediate, pos=0.66] {} 
          }
        edge from parent
          node [intermediate, pos=0.33] {} 
          node [intermediate, pos=0.66] {} 
      }
    }
    edge from parent
      node [intermediate, pos=0.33] {} 
      node [intermediate, pos=0.66] {} 
  };

\end{tikzpicture}
    \caption{First two layers of the tree from Figure $2$ with decoration number $g=2$; note this is also the Cayley tree with degree $d=2$ and decoration number $g=5$.}
\end{figure}
\end{subsubsection}
\end{subsection}
\end{section}
\begin{section}{Irregular Trees}
Following our second proof of the Cayley-tree result using part (g) of Theorem \ref{function}, we can extend the result to irregular trees by starting with a lemma:
    \begin{lemma}
        Let $d_i$ be a sequence of real numbers with $d_i\geq 2$. If the partial products \begin{equation}a_n=\prod_{i=1}^{n}\frac{3}{(d_i-1)}\end{equation} is summable, such that $\sum_{n}a_n<\infty$ then the infinite functional composition $\bigcirc_{i=1}^{\infty}F_{d_i}(x)>0$ for all $x\in (0,1]$. Moreover if the sequence of partial products is such that $a_n\to\infty$ then the infinite functional composition $\bigcirc_{i=1}^{\infty}F_{d_i}(x)=0$ for all $x\in[0,1]$.
    \end{lemma}
    \begin{proof}
        Note that from the bound $F_d(x)\geq \left(\frac{3}{(d-1)x}+1\right)^{-1}$ we have that the iterates \begin{align*}\lim_{N\to\infty}\bigcirc_{i=1}^{N}F_{d_i}(x)\geq \left(\frac{1}{x}\cdot\prod_{i=1}^{N}\frac{3}{d_i-1}+1+\sum_{k=1}^{N-1}\prod_{i=1}^{k}\frac{3}{d_i-1}\right)^{-1}=\lim_{N\to\infty}\left(\frac{1}{x}\cdot a_N+1+\sum_{k=1}^{N-1}a_k\right)^{-1}\\=\left(1+\sum_{k}a_k\right)^{-1}>0\end{align*} since $a_N\to 0$ and $\sum_k a_k<\infty$.\par
        For the second implication, notice that if the sequence $a_n\to\infty$ then \begin{equation}\lim_{n\to\infty}\bigcirc_{i=1}^{n}F_{d_i}(x)\leq\lim_{n\to\infty}x\cdot\prod_{i=1}^{n}\frac{d_i-1}{3}=\lim_{n\to\infty}x\cdot\frac{1}{a_n}=0\end{equation} for all $x\in[0,1]$
    \end{proof}
    \begin{corollary}\label{main}
        Let ${d_i}$ be a sequence of real numbers with $d_i\geq 2$; if \begin{equation}\ln(\mu):=\lim_{n\to\infty}\frac{1}{n}\sum_{i=1}^{n}\ln((d_i-1)/3)=\lim_{n\to\infty}\frac{1}{n}\ln(a_n)>0\end{equation} then the infinite function composition \begin{equation}\bigcirc_{i=1}^{\infty}F_{d_i}(x)>0\end{equation} for all $x\in (0,1]$, and
        if \begin{equation}\lim_{n\to\infty}\frac{1}{n}\sum_{i=1}^{n}\ln((d_i-1)/3)\leq 0\end{equation} then \begin{equation}\bigcirc_{i=1}^{\infty}F_{d_i}(x)=0\end{equation} for all $x\in[0,1]$
    \end{corollary}
    \begin{proof}
    Note that since $\lim_{n\to\infty}\frac{1}{n}\sum_{i=1}^{n}\ln((d_i-1)/3)=\ln(\mu)>0$ we have that for any $\epsilon>0$ we have that there exists an $N$ such that $\frac{1}{n}\sum_{i=1}^{n}\ln((d_i-1)/3)>(1-\epsilon)\ln(\mu)>0$ for all $n\geq N$; thus we have that $a_n\geq\mu^{(1-\epsilon)n}$; thus 
    \begin{align*}\lim_{N\to\infty}\bigcirc_{i=1}^{N}F_{d_i}(x)\geq \left(1+\sum_{k=1}^{N}\prod_{i=1}^{k}\frac{3}{d_i-1}+\sum_{k=N+1}^{\infty}\prod_{i=1}^{k}\frac{3}{d_i-1}\right)^{-1}\\\geq\left(1+\sum_{k=1}^{N}\prod_{i=1}^{k}\frac{3}{d_i-1}+\frac{1-\mu^{(\epsilon-1)(N+2)}}{1-\mu^{\epsilon-1}}\right)^{-1}>0\end{align*}
    Similarly note that if $\ln(\mu)\leq 0$ then $\frac{1}{a_n}\leq \mu^{(1-\epsilon)n}$ for large enough $n$, so we have \begin{equation}\lim_{N\to\infty}\bigcirc_{i=1}^{N}F_{d_i}(x)\leq \lim_{N\to\infty}\frac{1}{a_N}\cdot x=0\end{equation} for all $x\in [0,1]$
    \end{proof}
    Note that we do not have a uniform lower bound on the functional composition in the first part of the proof of Corollary \ref{main}. One tantalizing generalization is if we can uniformly lower bound this for trees that are not characterized by a single sequence $\{d_i\}$ of degrees uniformly radiating from the origin. However we have a counterexample to this.
    \begin{counterexample}
        For any $d\geq 5$ define the following tree via the sequence of degrees $\{d_i\}$ as $d_i=2$ for all $i\leq N$ and $d_i=5$ for all $i>N$. Then we have $\lim_{n\to\infty}\bigcirc_{i=1}^{n}F_{d_i}(x)=\frac{1}{(-3)^N}\cdot x_5$ where $x_5$ is the positive solution to $F_5(x_5)=-x_5$ which can be made as small as we like while maintaining that $\mu=\frac{4}{3}$.
    \end{counterexample} 
    This prevents us from using a simple macroscopic condition (such as all sequences obeying the $\mu$ condition) to obtain a bound for trees with arbitrary degrees; however we can still make this work for irregular trees with some added assumptions.
    \begin{theorem} Let $\Gamma$ be a tree with distinguished root $\overline{0}$ such that $\deg(x)=d_i$ for all $x$ such that $|x-\overline{0}|=i$. If \begin{equation}\ln(\mu)=\lim_{n\to\infty}\frac{1}{n}\sum_{i=1}^{n}\ln((d_i-1)/3)>0\end{equation} then the AKLT ground state on this tree is not unique.
    \end{theorem}
    \begin{proof}
    We note that this condition immediately implies that the infinite composition $\bigcirc_{i=1}^{\infty}F_{d_i}(x)>0$ for $x\in(0,1]$ by the previous lemma; let $N$ be such that for all $n\geq N$ we have that $\bigcirc_{i=1}^{n}F_{d_i}(x)=-x$ has at least one, solution; let $x_n$ be the smallest solution of $\bigcirc_{i=1}^{n}F_{d_i}(x)=-x$; this sequence is bounded on $(\epsilon,1]$ for some $\epsilon>0$ and so has a convergent subsequence $x_{n_k}\to x$ with $x\neq 0$. Let $\mathbf{x}_{n_k}=x_{n_k}\sigma_1$; then the ground state expectations of an observable $A$ with support on $\overline{0}$ of a sequence of boundary conditions \begin{align*}\lim_{k\to\infty}\langle\psi_{n_k}(B_{T_{n_k}}(\mathbf{x}_{n_k})),A\psi_{n_k}(B_{T_{n_k}}(\mathbf{x}_{n_k}))\rangle=\lim_{k\to\infty}\langle\psi_0(\tilde{\F}_{T_{n_k}}(B_{T_{n_k}}(\mathbf{x}_{n_k}))),A\psi_0(\tilde{\F}_{T_{n_k}}(B_{T_{n_k}}(\mathbf{x}_{n_k}))\rangle\\=\langle\psi_0(\id+x\sigma_1),A\psi_0(\id+x\sigma_1)\rangle\neq \lim_{k\to\infty}\langle\psi_{n_k}(B_{T_{n_k}}(\mathbf{0})),A\psi_{n_k}(B_{T_{n_k}}(\mathbf{0})\rangle=\langle\psi_0(\id),A\psi_0(\id)\rangle\end{align*}
    \end{proof}
    We can also use the above method to give a result for a general irregular tree with a simple condition.
    \begin{theorem}\label{maintree}
    Let $\Gamma$ be a tree such that there exists a $\mu>1$ and $C>0$ such that for every sequence $\{d_{i_k}\}$ such that $|d_{i_k}-\overline{0}|=k$ we have that $\prod_{k=1}^n \frac{d_{i_k}-1}{3}\geq C\mu^n$. Then the ground state of this model is not unique.
    \end{theorem}
    \begin{proof}
    For each $n$ let $\alpha_k$ be the index of the leaf corresponding to the minimal sequence \begin{equation}\min_{\{i_k\}}\left(\prod_{k=1}^n \frac{d_{i_k}-1}{3}\right)\end{equation} use the bound that $f_{d}(x_1,...,x_n)\geq f_d(x_\alpha)$ and thus the composition \begin{equation}f_{d_{i_0}}(f_{d_{i_1}}(f_{d_{i_2}}(...)))\geq \bigcirc_{k=1}^nf_{d_{\alpha_k}}(x)\end{equation} and from the above condition we know that from the bounds previously used that \begin{equation}\bigcirc_{k=1}^nf_{d_{\alpha_k}}(x)\geq \left(\sum_{m=1}^n\prod_{k=1}^m\frac{3}{d_{a_k}-1}+1\right)^{-1}\geq \left(C\sum_{m=1}^n\mu^{-m}+1\right)^{-1}\geq\left(\frac{C\mu^{-1}}{1-\mu^{-1}}+1\right)^{-1}>0\end{equation} Thus the concatenations of the transfer functions are uniformly bounded away from zero so there exists a convergent subsequence that is positive and one that is negative, so the ground state is not unique.
    \end{proof}
    Note that in this last proof we could not use the geometric mean condition for the minimal sequence, as the functional compositions associated to the set of such sequences is not uniformly bounded below as was shown in Counterexample 1.\par
    One can see the constant $3\mu$ as something of a "connective" constant in that it quantifies the exponential growth of all the self-avoiding paths from the root to the leaves in the tree.
\end{section}
\begin{section}{Bilayer trees}
The case of bilayer trees is significantly different, and requires a slightly different analysis; however the same general fixed-point type argument exists.
\subsection*{Bilayer Cayley tree}
\begin{figure}[t]
    \centering
    \tdplotsetmaincoords{70}{110}
\begin{tikzpicture}[tdplot_main_coords,scale=1.5]

    \def\r{2} 
    \def\h{3} 

    \coordinate (center1) at (0, 0, 0); 
    \foreach \i in {1, 2, 3, 4, 5} {
        \coordinate (lower\i) at ({\r*cos(\i*72)}, {\r*sin(\i*72)}, 0);
        \draw[thick] (center1) -- (lower\i);
        \node[circle, fill=black, inner sep=2pt] at (lower\i) {};
    }
    \node[circle, fill=black, inner sep=2pt] at (center1) {};

    \coordinate (center2) at (0, 0, \h); 
    \foreach \i in {1, 2, 3, 4, 5} {
        \coordinate (upper\i) at ({\r*cos(\i*72)}, {\r*sin(\i*72)}, \h);
        \draw[thick] (center2) -- (upper\i);
        \node[circle, fill=black, inner sep=2pt] at (upper\i) {};
    }
    \node[circle, fill=black, inner sep=2pt] at (center2) {};

    \draw[thick] (center1) -- (center2);
    \node[circle, fill=black, inner sep=2pt] at (center1) {};
    \node[circle, fill=black, inner sep=2pt] at (center2) {};

\end{tikzpicture}
    \caption{First layer of bilayer Cayley tree of degree $d=6$ which has splitting number $g=d-2=4$}
    \label{fig:bilayer}
\end{figure}
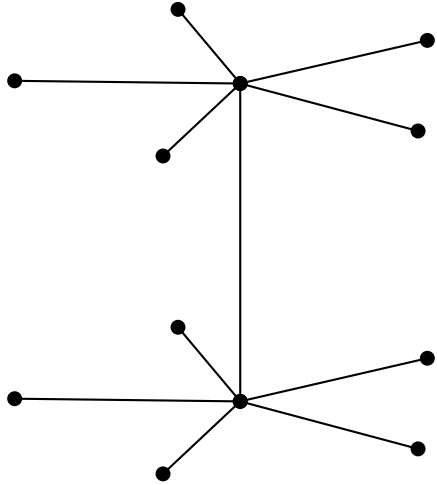

The bilayer tree is substantially different. This is because we have non-trivial operators which commute with representations of SU$(2)$. Note that in the bilayer tree, if we have the transfer operator at each (bilayer) node in the tree $\tilde{\F}:M_2(\mathbb{C})^{\otimes 2(d-1)}\to M_2(\mathbb{C})^{\otimes 2}$ we look for a pair of density matrices $\rho_0,\rho_1\in M_2(\mathbb{C})^{\otimes 2}$ such that $T(\rho_i^{\otimes (d-1)})=\rho_{1-i}$. Note that within $M_2(\mathbb{C})^{\otimes 2}$ we have the Casimir $\sum_i \sigma_i\otimes\sigma_i$ which commutes with the $SU(2)$ representation so we cannot have that $T(\mathbbm{1}\otimes\mathbbm{1})=\mathbbm{1}$; however $T$ is still trace-preserving. Thus if we define $\boldsymbol{\sigma}=[\sigma_{lk}]$ where $\sigma_{lk}=\sigma_l\otimes\sigma_k$ with $l,k\in\{0,1,2,3\}$; then our boundary conditions $B$ can be written as $\mathbbm{1}\otimes\mathbbm{1}+\mathbf{x}\cdot\boldsymbol{\sigma}$. We find fixed points of the transfer operator, the coefficients of which are rational functions of this vector.
For simplicity of notation, for a bilayer tree of degree $d$ we call $g-2$ the \textbf{splitting number} of the tree. 
\begin{theorem}
    The bilayer Cayley tree has a unique ground state for $g=1,2$.
\end{theorem}
\begin{proof}
The case of $g=1$ is the case of a bilayer AKLT chain which has a unique ground state \cite{aklt}. However, by way of an example, we prove this using a fixed point method. Note that the transfer operator can be fully diagonalized so that, starting with a boundary condition \begin{equation}B=\mathbbm{1}\otimes\mathbbm{1}+\sum_i x_{0i}\mathbbm{1}\otimes\sigma_i+\sum_ix_{i0}\sigma_i\otimes \mathbbm{1}+\sum_{i\neq j}x_{ij}\sigma_i\otimes\sigma_j+\sum_i x_{ii}\sigma_i\otimes\sigma_i=\mathbbm{1}\otimes\mathbbm{1}+\mathbf{x}\cdot \boldsymbol{\sigma}\end{equation} then we have \begin{equation}\tilde{\F}(B)=\mathbbm{1}\otimes\mathbbm{1}+\sum_i f_{i0}(\mathbf{x})(\sigma_i\otimes\mathbbm{1})+\sum_i f_{0i}(\mathbf{x})(\mathbbm{1}\otimes \sigma_i)+\sum_{i\neq j} f_{ij}(\mathbf{x})(\sigma_i\otimes\sigma_j)+\sum_i f_{ii}(\mathbf{x})(\sigma_i\otimes\sigma_i)\end{equation} where \begin{equation}f_{0i}(\mathbf{x})=\frac{-\frac{1}{3}x_{0i}-\frac{1}{9}x_{i0}}{\left(1+\frac{\sum_i x_{ii}}{9}\right)},\qquad
f_{i0}(\mathbf{x})=\frac{-\frac{1}{3}x_{i0}-\frac{1}{9}x_{0i}}{\left(1+\frac{\sum_i x_{ii}}{9}\right)}\end{equation}\begin{equation}
f_{ij}(\mathbf{x})=\frac{\frac{1}{9}x_{ij}}{\left(1+\frac{\sum_i x_{ii}}{9}\right)},\qquad
f_{ii}(\mathbf{x})=\frac{1+\frac{x_{ii}}{9}}{\left(1+\frac{\sum_i x_{ii}}{9}\right)}\end{equation}
we search for a fixed point $f_{kl}(\mathbf{x})=\pm x_{kl}$ and so we have the set of equations \begin{equation}\pm x_{0i}\left(1+\frac{\sum_i x_{ii}}{9}\right)=-\frac{1}{3}x_{0i}-\frac{1}{9}x_{i0},\qquad
\pm x_{i0}\left(1+\frac{\sum_i x_{ii}}{9}\right)=-\frac{1}{3}x_{i0}-\frac{1}{9}x_{0i}\end{equation}
\begin{equation}\pm x_{ij}\left(1+\frac{\sum_i x_{ii}}{9}\right)=\frac{1}{9}x_{ij},\qquad
\pm x_{ii}\left(1+\frac{\sum_i x_{ii}}{9}\right)=1+\frac{x_{ii}}{9}\end{equation} which has solution $x_{i0}=x_{0i}=x_{ij}=0$ and \begin{equation}x_{ii}=\frac{\left(\frac{1}{9}+\frac{x_{ii}}{9}\right)}{\left(1+\frac{\sum_i x_{ii}}{9}\right)}\end{equation} or \begin{equation}\left(\frac{1}{9}+\frac{\sum_i x_{ii}}{9}\right)=\left(\sum_i x_{ii}\right)\left(1+\frac{\sum_i x_{ii}}{9}\right)\end{equation} which has solution $\left(\sum_i x_{ii}\right)^2+8\sum_i x_{ii}-3=0$ and thus $\sum_i x_{ii}=-4+\sqrt{19}$ which is within the unit sphere. This value, while not being nonzero, is a fixed point of $\tilde{\F}$ that does not break the SU$(2)$ symmetry; moreover, the functions $f$ converge exponentially quickly to this fixed point, so we have a unique ground state.\par
For $g=2$, like before, we have our bilayer transfer operator $\tilde{E}:M_2(\mathbb{C})^{\otimes 4}\to M_2(\mathbb{C})^{\otimes 2}$ where we can explicitly calculate for $B=\mathbbm{1}\otimes\mathbbm{1}+\mathbf{x}\cdot \sigma$ that we have \begin{equation}\tilde{\F}(B)=\mathbbm{1}\otimes\mathbbm{1}+\sum_i f_{i0}(\mathbf{x})(\sigma_i\otimes\mathbbm{1})+\sum_i f_{0i}(\mathbf{x})(\mathbbm{1}\otimes \sigma_i)+\sum_{i\neq j} f_{ij}(\mathbf{x})(\sigma_i\otimes\sigma_j)+\sum_i f_{ii}(\mathbf{x})(\sigma_i\otimes\sigma_i)\end{equation} where each $f$ is a rational function of the vector $\mathbf{x}$. 
We have that we are solving the equation \begin{equation}\pm x_{kl}=f_{kl}(\mathbf{x})\end{equation} and so we get that we need to solve \begin{align*}\left(-\frac{2}{3}x_{i0}-\frac{2}{9}x_{0i}-\frac{4}{15}x_{ii}x_{i0}-\sum_{j\neq i}\frac{2}{45}x_{j0}y_{ij}-\sum_{j\neq i}\frac{2}{9}x_{i0}x_{ij}\right)\\\pm x_{i0}\left(1+\sum_{i}\frac{1}{3}x_{i0}^{2}+\frac{1}{9}x_{i0}x_{0i}+\sum_i\frac{2}{9}x_{ii}+\sum_{i}\frac{1}{9}x_{ii}^{2}+\sum_{i,j}\frac{1}{9}x_{ij}^{2}\right)=0\end{align*}
and
\begin{align*}\left(-\frac{2}{3}x_{0i}-\frac{2}{9}x_{i0}-\frac{4}{15}x_{ii}x_{0i}-\sum_{j\neq i}\frac{2}{45}x_{0j}y_{ji}-\sum_{j\neq i}\frac{2}{9}x_{0i}x_{ji}\right)\\\pm x_{0i}\left(1+\sum_{i}\frac{1}{3}x_{i0}^{2}+\frac{1}{9}x_{i0}x_{0i}+\sum_i\frac{2}{9}x_{ii}+\sum_{i}\frac{1}{9}x_{ii}^{2}+\sum_{i,j}\frac{1}{9}x_{ij}^{2}\right)=0\end{align*}
and
\begin{align*}\left(\frac{2}{9}x_{ij}+\frac{2}{9}x_{i0}x_{0j}+\frac{2}{45}x_{i0}x_{j0}+\frac{2}{75}x_{ij}x_{ii}+\frac{2}{75}x_{ij}x_{jj}+\frac{4}{225}x_{ij}x_{kk}+\frac{2}{225}x_{ik}x_{kj}\right)\\\pm x_{ij}\left(1+\sum_{i}\frac{1}{3}x_{i0}^{2}+\frac{1}{9}x_{i0}x_{0i}+\sum_i\frac{2}{9}x_{ii}+\sum_{i}\frac{1}{9}x_{ii}^{2}+\sum_{i,j}\frac{1}{9}x_{ij}^{2}\right)=0\end{align*}
and\begin{align*}\left(\frac{1}{9}+\frac{2}{9}x_{ii}+\frac{2}{9}x_{i0}x_{0i}+\frac{2}{15}x_{i0}^{2}+\sum_{j\neq i}\frac{2}{45}x_{i0}x_{j0}+\frac{1}{25}x_{ii}^{2}+\frac{2}{75}x_{ij}x_{ij}+\sum_{i,j}\frac{1}{225}x_{ij}^{2}+\sum_{i,j}\frac{2}{225}x_{ii}x_{jj}\right)\\\pm x_{ii}\left(1+\sum_{i}\frac{1}{3}x_{i0}^{2}+\frac{1}{9}x_{i0}x_{0i}+\sum_i\frac{2}{9}x_{ii}+\sum_{i}\frac{1}{9}x_{ii}^{2}+\sum_{i,j}\frac{1}{9}x_{ij}^{2}\right)=0\end{align*}
This set of equations only has the solution $\mathbf{x}=\mathbf{0}$ and since $f_ij'<1$ for all $i,j$ it has a unique limit. The same analysis as in \cite{fnw} can be done to show this ground state is unique.
\end{proof}
Next we show that we have non-uniqueness for splitting number $g=3$
\begin{theorem}
The AKLT Hamiltonian defined on the $g=3$ bilayer Cayley tree does not have a unique ground state.
\end{theorem}
\begin{proof}
We start with the ansatz that we have a boundary condition of the form
\begin{equation}B(\mathbf{x})=\mathbbm{1}\otimes\mathbbm{1}+\mathbf{x}\cdot\boldsymbol{\sigma}\end{equation} where we note that the space of boundary conditions such that $x_1=x_{0i}=x_{i0}$ and $x_2=x_{ij}=x_{kl}$ and $x_3=x_{ii}=x_{jj}$ is a fixed subspace, so we assume $\mathbf{x}$ satisfies these and search for two boundary conditions $B_+$ and $B_-$ that are solutions to \begin{equation}\tilde{\F}(B_{\pm}(\mathbf{x})^{\otimes 3})=B_{\mp}(\mathbf{x})\end{equation} which means \begin{equation}f_{0i}(\mathbf{x})=f_1(x_1,x_2,x_3)=-x_1\end{equation}
\begin{equation}f_{ij}(\mathbf{x})=f_2(x_1,x_2,x_3)=x_2\end{equation}
\begin{equation}f_{ii}(\mathbf{x})=f_3(x_1,x_2,x_3)=x_3\end{equation}
where $f_1,f_2,f_3$ are rational functions; define $f_0(\mathbf{x})=\Tr(\F(B_{\pm}(\mathbf{x})^{\otimes 3}))$, so we can rearrange by searching for a root of the following polynomials:
\begin{align*}f_1(\mathbf{x})+x_1\cdot f_0(\mathbf{x})=-3x_1 - \frac{37}{15}x_1^3 - \frac{2}{5}x_1x_2^2 - \frac{8}{15}x_1x_3 - \frac{4}{25}x_1x_2x_3 - \frac{14}{75}x_1x_2^2 + \frac{8}{75}x_1x_3^2\\ + \frac{2}{25}x_1x_3^{2} + \frac{22}{25}x_1x_2x_3 + \frac{3}{25}x_3^{3} + \frac{4}{75}x_2^{3} + \frac{8}{25}x_2^{2} + 2x_2^2 + 2x_3 + 6x_1^2 = 0\end{align*}
and \begin{align*}
    f_2(\mathbf{x})-x_2\cdot f_0(\mathbf{x})=-\frac{2}{3} + \frac{7}{5}x_1^2 + \frac{14}{75}x_2^2 + \frac{6}{25}x_3 + \frac{1}{15}x_3^2 + \frac{11}{75}x_2x_3\\ + \frac{22}{75}x_2x_1^2 + \frac{4}{25}x_3x_1^2 + \frac{4}{25}x_2x_1^2 + \frac{4}{75}x_3x_2^2 + \frac{2}{25}x_2^3 + \frac{1}{25}x_3^3 = 0
\end{align*}
and
\begin{equation}f_3(\mathbf{x})-x_3\cdot f_0(\mathbf{x})=-\frac{8}{9} + \frac{142}{75}x_1^{2} + \frac{8}{25}x_2^{2} + \frac{1}{3}x_3 + \frac{13}{75}x_3^{2} + \frac{1}{15}x_3^{3} + \frac{16}{75}x_2x_1^{2} + \frac{8}{25}x_3x_2^{2} = 0
\end{equation}
From a simple check with a graphing calculator we find there is a pair of solutions $\mathbf{x}_{\pm}\approx[\pm0.3020,0.0466,0.1754]$ such that $\tilde{\F}(B(\mathbf{x}_{\pm})^{\otimes 3})=B(\mathbf{x}_{\mp})$.
\end{proof}
Note that while this is the same degree as the single-layer Cayley tree with degree $d_s=4$ with splitting number $g_s:=d_s-1=3$, the bilayer degree $d_b=5$ Cayley tree has the same splitting number $g_b=d_b-2=3$; however, in the single-layer Cayley tree we have a unique ground state, whereas in the bilayer case we have a degenerate N\'{e}el ordered ground state space.\cite{fnw}. 
\end{section}
\begin{section}{Acknowledgements}
The author would like to thank Bruno Nachtergaele for his guidance and support. This research was supported in part by NSF Grants DMS-2108390 and DMS-2510824.
\end{section}
\bibliographystyle{plain}
\bibliography{bib1.bib}
\end{document}